\documentclass[conference,compsoc]{IEEEtran}

\usepackage{multirow}
\usepackage{nicefrac}
\usepackage{color}
\usepackage{amsmath}
\usepackage{graphicx}
\usepackage{psfrag}
\usepackage{url}
\usepackage{wasysym}
\usepackage{amsmath}
\usepackage{amsfonts}
\usepackage{amssymb}
\usepackage{tikz}

\usepackage{hyperref}

\usepackage{subcaption}
\usepackage{enumerate}

\usepackage{amsthm}

\newtheorem{lemma}{Lemma}

\newtheorem{corollary}{Corollary}

\newcommand{\on}{\ensuremath{\mathsf{on}}}
\newcommand{\off}{\ensuremath{\mathsf{off}}}

\newcommand{\ton}[1]{\ensuremath{t_{#1}^{\mathsf{on}}}}
\newcommand{\toff}[1]{\ensuremath{t_{#1}^{\mathsf{off}}}}
\newcommand{\qmax}[1]{\ensuremath{q_{#1}^\mathsf{max}}}
\newcommand{\mbar}[1]{\ensuremath{\bar{m}_{#1}}}

\newcommand{\qhat}[1]{\ensuremath{\hat{q}_{#1}}}
\newcommand{\ns}[1]{\ensuremath{\bar{s}_{#1}}}
\newcommand{\floor}[1]
{\left\lfloor #1 \right\rfloor}
\newcommand{\partder}[2]
{\frac{\partial #1}{\partial #2}}
\newcommand{\delay}[2]{\ensuremath{D_{#1,#2}}}
\newcommand{\maxdelay}[2]{\ensuremath{\hat{D}_{#1,#2}}}

\renewcommand{\paragraph}[1]{\hspace{2ex}\textbf{#1}.\hspace{2ex}}

 \usetikzlibrary{shapes,positioning}
 \usetikzlibrary{arrows,shapes,automata,petri,calc,fit}
\usepackage{pgfplots}
\usetikzlibrary{calc}
\pgfplotsset{compat=newest}
\usetikzlibrary{fpu}

\makeatletter
\renewcommand{\section}{\@startsection{section}{1}{\z@}%
{0.6\baselineskip plus 0.25\baselineskip minus 0.2\baselineskip}
{0.4\baselineskip plus 0.25\baselineskip minus 0.2\baselineskip}
{\normalfont\large\bfseries}}%

\renewcommand{\subsection}{\@startsection{subsection}{2}{\z@}%
{0.6\baselineskip plus 0.25\baselineskip minus 0.2\baselineskip}
{0.4\baselineskip plus 0.25\baselineskip minus 0.2\baselineskip}
{\normalfont\sublargesize\bfseries}}%
\makeatother

\title{Cost minimization of network services \\ with buffer and
  end-to-end deadline constraints}

\author{ \IEEEauthorblockN{Victor Millnert\IEEEauthorrefmark{1}, Johan
    Eker\IEEEauthorrefmark{1}\IEEEauthorrefmark{2}, Enrico
    Bini\IEEEauthorrefmark{3}}
  \IEEEauthorblockA{\IEEEauthorrefmark{1}Lund University, Sweden}
  \IEEEauthorblockA{\IEEEauthorrefmark{2}Ericsson Research, Sweden}
  \IEEEauthorblockA{\IEEEauthorrefmark{3}Scuola Superiore Sant'Anna,
    Pisa, Italy} }

\begin{document}

\maketitle

\begin{abstract}
  Cloud computing technology provides the means to share physical
  resources among multiple users and data center tenants by exposing
  them as virtual resources.  There is a strong industrial drive to
  use similar technology and concepts to provide timing sensitive
  services. One such is virtual networking services, so called
  services chains, which consist of several interconnected virtual
  network functions. This allows for the capacity to be scaled up and
  down by adding or removing virtual resources. In this work, we
  develop a model of a service chain and pose the  dynamic
  allocation of resources as an optimization problem. We design and
  present a set of strategies to allot virtual network nodes in an
  optimal fashion subject to latency and buffer constraints.
\end{abstract}

\section{Introduction}
\label{sec:intro}
Over the last years, cloud computing has swiftly transformed the
IT infrastructure landscape, leading to large cost-savings for
deployment of a wide range of IT applications. Some main
characteristics of cloud computing are resource pooling, elasticity,
and metering. Physical resources such as compute nodes, storage nodes,
and network fabrics are shared among tenants. Virtual resource
elasticity brings the ability to dynamically change the amount of
allocated resources, for example as a function of workload or cost.
Resource usage is metered and in most pricing models the tenant only
pays for the allocated capacity.

While cloud technology initially was mostly used for IT applications,
e.g. web servers, databases, etc., it is rapidly finding its way into
new domains. One such domain is processing of network packages. Today
network services are packaged as physical appliances that are
connected together using physical network. Network services consist of
interconnected network functions (NF). Examples of network functions
are firewalls, deep packet inspections, transcoding, etc. A recent
initiative from the standardisation body ETSI (European
Telecommunications Standards Institute) addresses the standardisation
of virtual network services under the name Network Functions
Virtualisation (NFV)~\cite{ETSI-NFV}. The expected benefits from this
are, among others, better hardware utilisation and more flexibility,
which translate into reduced capital and operating expenses (CAPEX and
OPEX).  A number of interesting use cases are found
in~\cite{ETSI-NFV-UC}, and in this technical report we are investigating the one
referred to as Virtual Network Functions Forwarding Graphs, see
Figure~\ref{fig:NFV-VNF}.
\begin{figure}[htb]
  \centering
  \includegraphics[width=\columnwidth]{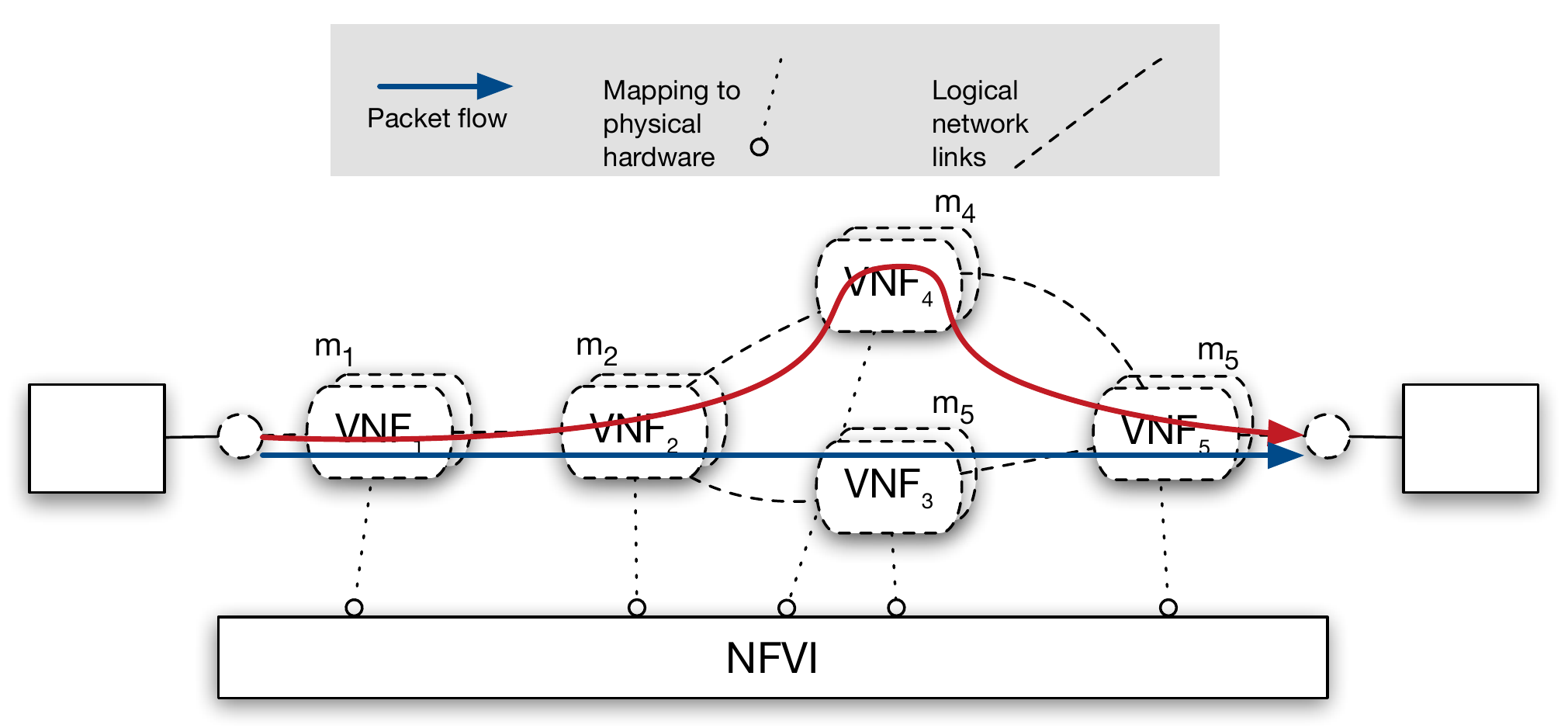}
  \caption{Several virtual networking functions (VNF) are connected
    together to provide a set of services. A packet flow is a specific
    path through the VNFs. Connected VNFs are referred to as virtual
    forwarding graphs or service chains. The VNFs are mapped onto
    physical hardware, i.e. compute nodes and network fabrics and this
    underlying hardware infrastructure is referred to as
    NFVI.}
  \label{fig:NFV-VNF}
\end{figure}

We investigate the allocation of virtual resources to a given packet
flow, i.e. what is the most cost efficient way to allocate VNFs with a
given capacity that still provide a network service within a given
latency bound? The distilled problem is illustrated as the packet
flows in Figure~\ref{fig:NFV-VNF}. The forwarding graph is implemented
as a chain of virtual network nodes, also known as a service
chains. To ensure that the capacity of a service chain matches the
time-varying load, the number of instances $m_i$ of each individual
network function $\text{VNF}_i$ may be scaled up or down.

The contribution of the technical report is
\begin{itemize}
\item a mathematical model of the virtual resources supporting the
  packet flows in Figure~\ref{fig:NFV-VNF},
\item the set-up of an optimization problem for controlling the number
  of machines needed by each function in the service chain,
\item solution of the optimization-problem leading to a control-scheme
  of the number of machines needed to guarantee that the end-to-end
  deadline is met for incoming packets under a constant input flow.
\end{itemize}

\subsection*{Related works}
\label{sec:related}

There are a number of well known and established resource management
frameworks for data centers, but few of them explicitly address the
issue of latency.  Sparrow~\cite{Ousterhout:2013:SDL:2517349.2522716}
presents an approach for scheduling a large number of parallel jobs
with short deadlines. The problem domain is different compared to our
work in that we focus on sequential rather than parallel
jobs. Chronos~\cite{Kapoor:2012:CPL:2391229.2391238} focuses on
reducing latency on the communication
stack. RT-OpenStack~\cite{xi2015rt} adds real-time performance to
OpenStack by usage of a real-time hypervisor and a timing-aware
VM-to-host mapping.

The enforcement of an end-to-end (E2E) deadline of a sequence of jobs
to be executed through a sequence of computing elements was addressed
by several works, possibly under different terminologies. In the
holistic analysis~\cite{Tin94,Pal03,Pel07} the schedulability analysis
is performed locally. At global level the local response times are
transformed into jitter or offset constraints for the subsequent tasks.

A second approach to guarantee an E2E deadline is to split a
constraint into several local deadline constraints. While this
approach avoids the iteration of the analysis, it requires an
effective splitting method. Di Natale and Stankovic~\cite{DiN94}
proposed to split the E2E deadline proportionally to the local
computation time or to divide equally the slack time.  Later, Jiang
\cite{Jia06} used time slices to decouple the schedulability analysis
of each node, reducing the complexity of the analysis. Such an
approach improves the robustness of the schedule, and allows to
analyse each pipeline in isolation. Serreli et al.~\cite{Ser09,Ser10}
proposed to assign local deadlines to minimize a linear upper bound of
the resulting local demand bound functions. More recently, Hong et
al~\cite{Hon15} formulated the local deadline assignment problem as a
MILP with the goal of maximising the slack time.  After local
deadlines are assigned, the processor demand criterion was used to
analyze distributed real-time pipelines~\cite{Rah08,Ser10}.

In all the mentioned works, jobs have non-negligible execution
times. Hence, their delay is caused by the preemption experienced at
each function.  In our context, which is scheduling of virtual network
services, jobs are executed non-preemptively and in FIFO order. Hence,
the impact of the local computation onto the E2E delay of a request is
minor compared to the queueing delay.  This type of delay is
intensively investigated in the networking community in the broad area
\emph{queuing systems}~\cite{Kle75}.  In this area, Henriksson et
al.~\cite{Hen04} proposed a feedforward/feedback controller to adjust
the processing speed to match a given delay target.

Most of the works in queuing theory assumes a stochastic (usually
markovian) model of job arrivals and service times. A solid
contribution to the theory of deterministic queuing systems is due to
Baccelli et al.~\cite{Bac92}, Cruz~\cite{Cru91}, and Parekh \&
Gallager~\cite{Par93}. These results built the foundation for the
\emph{network calculus}~\cite{LeB01}, later applied to real-time
systems in the \emph{real-time calculus}~\cite{Cha05}.  The advantage
of network/real-time calculus is that, together with an analysis of
the E2E delays, the sizes of the queues are also modelled. As in the
cloud computing scenario the impact of the queue is very relevant
since that is part of the resource usage which we aim to minimize,
hence we follow this type of modeling.

\section{Problem formulation}
\label{sec:model}

To analyse the resource management problem described in
Section~\ref{sec:intro}, we model an abstract version of
Figure~\ref{fig:NFV-VNF} with the one shown in
Figure~\ref{fig:service-chain}. In our model we consider each VNF
simply as a \emph{function} that is processing requests. Within each
function there are a number of \emph{machines} running (which in
Section~\ref{sec:intro} would correspond to virtual machines).

\begin{figure}[htb]
  \centering
  \includegraphics[width=\columnwidth]{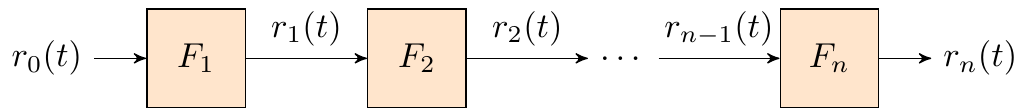}
  \caption{Illustration of the service-chain.}
  \label{fig:service-chain}
\end{figure}

\subsection{Input model}
\label{sec:input-model}

The service chain is composed by $n$ service functions. The $i$-th
function, denoted by $F_i$, receives requests at an \emph{incoming
  rate} $r_{i-1}(t)$. Then, the \emph{cumulative arrived requests} is
\begin{equation}
  R_{i-1}(t) = \int_0^t r_{i-1}(\tau)\, d\tau.
  \label{eq:intR}
\end{equation}

We model incoming requests and service speeds of each functions by a
fluid approximation.  In fact, in \cite{Zhang:aa} they used recent
advances in NFV-technology to process requests with a throughput of
about 10 million requests per second. We believe this to show that the
possible discretization error when using a fluid approximation is
indeed negligible.

Finally, each request needs to pass through the entire service-chain
within an end-to-end deadline, denoted $D^\mathsf{max}$.

\subsection{Service model}
\label{sec:service}

As illustrated in Figure~\ref{fig:service-chain-expanded}, the
incoming requests to function $F_i$ are stored in the queue and then
processed once it reaches the head of the queue. Here one should note
that due to the fluid approximation we made earlier, our analysis will
assume that a request is processed in parallel by all present machines
in the function. Again, with the requests entering at a rate of
millions per second along with them being very small we believe that
this is a good abstraction. At time $t$ there are $m_i(t)$ machines
ready to serve the requests, each with a \emph{nominal speed} of
$\ns{i}$ (note that this nominal speed might differ between different
functions in the service chain, i.e. it does not in general hold that
$\ns{i}=\ns{j}$ for $i\neq j$ ). The \emph{maximum speed} that
function $F_i$ can process requests at is thus $m_i(t)\ns{i}$. The
rate by which $F_i$ is processing requests is denoted $s_i(t)$. The
\emph{cumulative served requests} is defined as
\begin{align}
  \label{eq:cumServReq}
   S_i(t) = \int_0^t s_i(\tau)\, d\tau.
\end{align}

\begin{figure}[htb]
  \centering
  \includegraphics[width=\columnwidth]{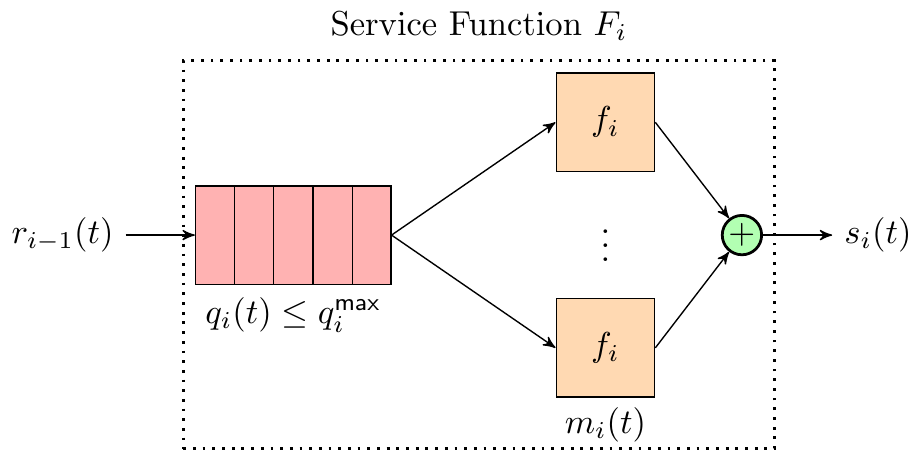}
  \caption{Illustration of the structure and different entities of the
    service chain.}
  \label{fig:service-chain-expanded}
\end{figure}

At time $t$ the number of requests stored in the queue is defined as
the \emph{queue length} $q_i(t)$:
\begin{equation}
  \label{eq:def-queue}
  q_i(t) = \int_0^t r_{i-1}(\tau) - s_i(\tau) d\tau = R_{i-1}(t) - S_i(t).
\end{equation}
Each function has a fixed \textit{maximum-queue capacity} $\qmax{i}$,
representing the largest number of requests that can be stored at the
function $F_i$. 

The \emph{queueing delay}, depends on the status of
the queue as well as on the service rate. We denote by
$\delay{i}{j}(t)$ the time taken by a request from when it enters
function $F_i$ to when it exits $F_j$, with $j\geq i$, where $t$ is
the time when the request exits function $F_j$:
\begin{equation*}
  \delay{i}{j}(t) = \inf\, \{\tau \geq 0 \,:\, R_{i-1}(t-\tau) \leq S_j(t)\}.
\end{equation*}
The \emph{maximum queueing delay} 
then is $\maxdelay{i}{j} = \max_{t\geq0} \delay{i}{j}(t)$.  The
requirement that a requests meets it end-to-end deadline is
$\maxdelay{1}{n}\leq D^\mathsf{max}$.

To control the queueing delay, it is necessary to control the service
rate of the function. Therefore, we assume that it is possible to
change the maximum service-rate of a function by changing the number
of machines that are on, i.e. changing $m_i(t)$. However, turning on a
machine takes $\Delta_i^\on$ time units, and turning off a machine
takes $\Delta_i^\off$ time units. Together they account for a
\emph{time delay}, $\Delta_i = \Delta_i^\on + \Delta_i^\off$,
associated with turning on/off a machine.

In the famous paper \cite{Kapoor:2012:CPL:2391229.2391238}, Google
profiled where the latency in a data center occurred. They showed that
less than 1\% (<1$\mu s$) of the latency occurred was due to the
propagation in the network fabric. The other 99\% ($\approx 85 \mu s$)
occurred somewhere in the kernel, the switches, the memory, or the
application. Since it is very difficult to say exactly which of this
99\% is due to processing, or queueing, we make the abstraction of
considering queueing delay and processing delay together, simply as
queueing delay. Hence, once a request has reached the head of the
queue and is processed it immediately exits the function and enters
the next function in the chain, or exit the chain if exiting the final
function. We thus assume that no request is lost in the communication
links, and that there is no propagation delay. Therefore, the
concatenation of the functions $F_1$ through $F_n$ implies that the
input of function $F_i$ is exactly the output of function $F_{i-1}$,
for $i=2,\ldots,n$, as illustrated in Figure~\ref{fig:service-chain}.

\subsection{Cost model}
\label{sec:costModel}

To be able to provide guarantees about the behaviour of the service
chain, it is necessary to make \emph{hard reservations} of the
resources needed by each function in the chain. This means that when a
certain resource is reserved, it is guaranteed to be available for
utilisation. Reserving this resource results in a cost, and due to the
hard reservation, the cost is not dependent on the actual utilisation,
but only on the resource reserved.

The \emph{computation cost} per time-unit per machine is denoted
$\jmath_i^\mathsf{c}$, and can be seen as the cost for the CPU-cycles
needed by one machine in $F_i$. This cost will also occur during the
time-delay $\Delta_i$. Without being too conservative, this time-delay
can be assumed to occur only when a machine is started. The
\emph{average computing cost} per time-unit for the whole function
$F_i$ is then
\begin{align}
  \label{eq:computeCost} 
  J_i^\mathsf{c}(m_i(t)) = \lim_{t \rightarrow \infty} \frac{\jmath_i^\mathsf{c}}{t} \int \limits_0^t m_i(s)
  +  \Delta_i\cdot(\partial_- m_i(s) )_+  ds
\end{align}
where $(x)_+ = \max(x,0)$, and $\partial_-m_i(t)$ is the left-limit of
$m_i(t)$:
\begin{equation*}
  \partial_- m_i(t) = \lim_{a \rightarrow t^-} \frac{m_i(t) -
  m_i(a)}{t - a},
\end{equation*}
that is, a sequence of Dirac's deltas at all points where the number of
machines changes.  This means that the value of the left-limit of
$m_i(t)$ is only adding to the computation-cost whenever it is
positive, i.e. when a machine is switched on.

The \emph{queue cost} per time-unit per space for a request is denoted
$\jmath_i^\mathsf{q}$ and can be seen as the cost for having a queue
with the capacity of one request. This cost come from the fact that
physical storage needs to be reserved so that a queue can be hosted on
it, normally this would correspond to the RAM of the network-card.
Reserving the capacity of $\qmax{i}$ would thus result in a cost per
time-unit of
\begin{align}
  \label{eq:bufferCost}
  J_i^\mathsf{q}(\qmax{i}) = \jmath_i^\mathsf{q} \qmax{i}.
\end{align}

\subsection{Problem definition}
\label{sec:optService}

The aim of this technical report is to control the number $m_i(t)$ of machines
running at stage $i$, such that the total average cost is minimized,
while the E2E constraint $D^\mathsf{max}$ is not violated and the
maximum queue sizes $\qmax{i}$ are not exceeded. This can be posed as
the following problem:
\begin{equation}
\begin{aligned}
  \text{minimize }  
  &  J = \sum_{i=1}^n J_i^\mathsf{c}(m_i(t))+J_i^\mathsf{q}(\qmax{i}) \\
  \text{subject to }
  &  \maxdelay{1}{n}\leq D^\mathsf{max} \\
  &  q_i(t) \leq \qmax{i}, \quad \forall t\geq0,\quad i=1,2,\ldots,n
\end{aligned}
\label{eq:optProblem}
\end{equation}
with $J_i^\mathsf{c}$ and $J_i^\mathsf{q}$ as
in~(\ref{eq:computeCost}) and~(\ref{eq:bufferCost}), respectively.  In
this technical report the optimization problem \eqref{eq:optProblem} will be
solved for a service-chain fed with a constant incoming rate $r$.

A valid lower bound $J^\mathsf{lb}$ to the cost achieved by any
feasible solution of~(\ref{eq:optProblem}) is found by assuming that
all functions are capable of providing exactly a service rate $r$
equal to the input rate. This is possible by running a fractional
number of machines $r/\ns{i}$ 
at function $F_i$.  In such an ideal case, buffers can be of zero size
($\forall i,\ \qmax{i}=0$), and there is no queueing delay
($\maxdelay{1}{n}=0$) since service and the arrival rates are the same
at all functions. Hence, the lower bound to the cost is
\begin{equation}
  \label{eq:costLowB}
  J^\mathsf{lb} = \sum_{i=1}^n\jmath_i^\mathsf{c}\frac{r}{\ns{i}}.
\end{equation}
Such a lower bound will be used to compare the quality of the several
solutions found later on.

In Section~\ref{sec:switch} we are going to make a general
consideration about the on/off scheme of each machine, in presence
of a constant input rate $r$. Later in Sections~\ref{sec:linear}
and~\ref{sec:periodic}, the optimal design problem
of~(\ref{eq:optProblem}) is solved, under a different set of
assumptions.

\section{Machine switching scheme}
\label{sec:switch}

In presence of an incoming flow of requests at a constant rate
$r_0(t)=r$, a number
\begin{equation}
  \mbar{i} = \floor{\frac{r}{\ns{i}}}
  \label{eq:mbarDef}
\end{equation}
of machines running the function $F_i$ must always stay on.  To match
the incoming rate $r$, in addition to the $\mbar{i}$ machines always
on, another machine must be on for some time in order to process a
request rate of $\ns{i}\rho_i$ where $\rho_i$ is the \emph{normalized
  residual request} rate:
\begin{equation}
  \label{eq:residual}
  \rho_i = \nicefrac{r}{\ns{i}} - \mbar{i},
\end{equation}
where $\rho_i\in[0,1)$.

In our scheme, the extra machine is switched on at a \emph{desired
  on-time} $\ton{i}$:
\begin{itemize}
\item $\off \rightarrow \on$: function $F_i$ switches on the
  additional machine when the time $t$ exceeds $\ton{i}$.
\end{itemize}
Since the additional machine does not need to always be on, it could
be switched off after some time.  The off-switching is also based on a
time-condition, the \emph{desired stop-time} $t_i^\off$,
i.e. the time-instance that the machine should be switched off, and is
given by:
\begin{equation*}
  t_i^\off = t_i^\on + T_i^\on.
\end{equation*}
where $T_i^\on$ is the duration that the machine should be on for, and
something that needs to be found.  The off-switching is then
triggered in the following way:
\begin{itemize}
\item $\on \rightarrow \off$: function $F_i$ switches off the
  additional machine when the time $t$ exceeds $t_i^\off$.
\end{itemize}

Note that this control-scheme, in addition with the constant input,
result in the extra machine being switched on/off
\emph{periodically}, with a period $T_i$. We thus assume that the
extra machine can process requests for a time $T_i^\on$ every period
$T_i$. The time during each period where the machine is not processing
any requests is denoted $T_i^\off=T_i-T_i^\on$. Notice, however, that
the actual time the extra machine is consuming power is
$T_i^\on+\Delta_i$ due to the time delay.

In the presence of a constant input, it is straight-forward to find
the necessary on-time during each period---in order for the additional
machine to provide the residual processing capacity of
$r-\mbar{i}\ns{i}$, its on-time $T_i^\on$ must be such that
\[
  T_i^\on\ns{i}=T_i(r-\mbar{i}\ns{i}),
\]
which implies
\begin{equation}
  \label{eq:TiON}
  T_i^\on = T_i\rho_i, \quad
  T_i^\off = T_i-T_i^\on
  = T_i( 1 - \rho_i).
\end{equation}

With each additional machine being switched on/off periodically,
it is also straightforward to find the computation cost for each
function. If $\mbar{i}+1$ machines are on for a time $T_i^\on$, and
only $\mbar{i}$ machines are on for a time $T_i^\off$, then the cost
$J_i^\mathsf{c}$ of~\eqref{eq:computeCost} becomes
\begin{equation}
  \label{eq:costCperiod}
  J_i^\mathsf{c} = \jmath_i^\mathsf{c}\left(\frac{T_i^\mathsf{on}+\Delta_i}{T_i}+\mbar{i}\!\right)
  = \jmath_i^\mathsf{c}\left(\!\mbar{i}+\rho_i+\frac{\Delta_i}{T_i}\right)
\end{equation}
if $T_i^\off\geq \Delta_i$. If instead $T_i^\off<\Delta_i$, that is if
\begin{equation}
  \label{eq:condToffSmall}
  T_i < \overline{T}_i := \frac{\Delta_i}{1-\rho_i},
\end{equation}
then there is no time to switch the additional machine off and then on
again. Hence, we keep the last machine on, even if it is not
processing packets, and the computing cost becomes
\begin{equation}
  \label{eq:castCsmallToff}
  J_i^\mathsf{c} = \jmath_i^\mathsf{c}\left(\!\mbar{i}+\rho_i+\frac{T_i^\off}{T_i}\right)
  = \jmath_i^\mathsf{c}(\mbar{i}+1).
\end{equation}

Next, using this control-scheme, the optimization problem of
\eqref{eq:optProblem} will be studied under two different set of
assumptions. In Section~\ref{sec:linear}, we will approximate the service
functions with linear lower-bounds, which allows us to find a period
$T_i$ of each function. Note that the lower-bound approximation incurs
in some pessimism in the solution. In Section~\ref{sec:periodic} we will
assume that every function will switch on/off its additional machine
with the same period, $T$. For this case we will derive the optimal
period $T$.

\section{Linear approximation of service}
\label{sec:linear}

In this section, the service functions are approximated by linear
lower-bounds. This choice allows us finding an explicit solution to
the switching periods $T_i$ of each function. Inevitably, the solution
incurs in some pessimism due to the approximation.

If the cumulative served requests~\eqref{eq:cumServReq} is
lower-bounded by a linear function, as illustrated in
Figure~\ref{fig:linear-approximation}, the maximum size of the
queue at function $F_i$ is attained exactly when the function switches on its extra machine,
$\qmax{i}=q_i(\ton{i})$:
\begin{equation}
  \label{eq:maxQi}
  \qmax{i}  = (r - \mbar{i}\ns{i}) T_i^\off
  = \ns{i}T_i\rho_i(1 - \rho_i),
\end{equation}
while the maximum introduced delay is
\begin{equation*}
  \maxdelay{i}{i}  = \frac{\ns{i}}{r} T_i\rho_i(1 -\rho_i) = \frac{q_i(\ton{i})}{r}.
\end{equation*}

\begin{figure}[ht]
  \centering
  \includegraphics[scale=1]{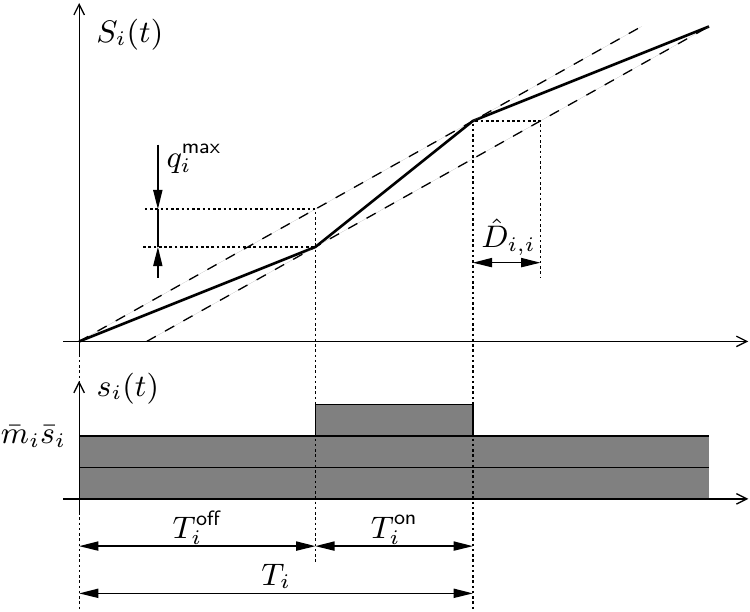}
  \caption{Linear approximation to the cumulative
    served requests.}
  \label{fig:linear-approximation}
\end{figure}

By setting the variable $x_i$ and constants $a_i$, $b_i$, and $c$ as
\begin{equation}
  \left\{
  \begin{aligned}
    x_i & = \maxdelay{i}{i} = \frac{\ns{i}}{r} T_i\rho_i(1 -\rho_i), &
    a_i & = \jmath_i^\mathsf{q}r,\\
    b_i & = \jmath_i^\mathsf{c}\Delta_i\frac{\ns{i}}{r} \rho_i(1 -\rho_i), &
    c & = D^\mathsf{max},
  \end{aligned}
  \right.
  \label{eq:setABCDXlin}
\end{equation}
the optimal design problem of~\eqref{eq:optProblem} can be formulated
as
\begin{align}
  \text{minimize } & J = \sum_{i=1}^n \Big(a_ix_i+b_i\frac{1}{x_i}\Big) +J^\mathsf{lb} \notag\\
  \text{such that } & \sum_{i=1}^n x_i \leq c
  \label{eq:constrSumLeqC}\\
  & x_i \geq 0 \notag
\end{align}
with $J^\mathsf{lb}$ being the cost lower bound as
in~(\ref{eq:costLowB}).
First, we check the unconstrained solution, which is
\begin{equation}
  \partder{J}{x_i} = 0 \quad\Rightarrow\quad
  x_i = \sqrt{\frac{b_i}{a_i}}
  \label{eq:solUncontr}
\end{equation}
If constraint \eqref{eq:constrSumLeqC} holds at the solution
of~(\ref{eq:solUncontr}), the optimum is unconstrained and the
corresponding optimal cost is
\begin{equation*}
  J' = 2\sum_{i=1}^n \sqrt{a_ib_i}+J^\mathsf{lb}.
\end{equation*}
Otherwise, the constraint \eqref{eq:constrSumLeqC} must be explicitly
enforced. In this case the solution is found via Lagrange
multiplier. Let $\lambda$ be the multiplier of the
constraint~(\ref{eq:constrSumLeqC}), then the solution is
\begin{equation}
  \label{eq:solConstr}
  x_i=\sqrt{\frac{b_i}{a_i+\lambda}}
\end{equation}
with cost
\begin{equation}
  \label{eq:costConstr}
  J'' = \sum_{i=1}^n \sqrt{a_ib_i}\left(\!\sqrt{\frac{a_i}{a_i+\lambda}}
    +\sqrt{\frac{a_i+\lambda}{a_i}}\right) +J^\mathsf{lb}\geq J'
\end{equation}
with the multiplier $\lambda$ being the unique positive solution of
\begin{equation}
  \label{eq:lambdaConstr}
  \sum_{i=1}^n\sqrt{\frac{b_i}{a_i+\lambda}} - c = 0
\end{equation}
Finally, the switching-period $T_i$ is given by
\[
T_i = \frac{r}{\ns{i}\rho_i(1-\rho_i)}x_i
\]
and the maximum queue-size $\qmax{i}$ are given by
Eq.\eqref{eq:maxQi}. Notice that, for all $i$ such that
Eq.~(\ref{eq:condToffSmall}) holds true, then there is physically no
time to switch the additional machine off and then on again
($T_i^\off < \Delta_i$). For all these machines the cost is computed
as $\mbar{i}+1$ machines are always on (as in
Eq.~(\ref{eq:castCsmallToff})) and not by~(\ref{eq:costCperiod}).

\paragraph{Example}
Let us apply the described design methodology to a simple example of a
service chain with two functions.  We assume an incoming rate $r=17$
of requests per second with an E2E-deadline of
$D^{\mathsf{max}} = 0.02$. The parameters of the functions are
reported in Table~\ref{tab:exampleParam}.
\begin{table}[tb]
  \centering
    \begin{tabular}{c|ccccc}
    \hline
    \hline
    $i$ & $\ns{i}$ & $\jmath_i^\mathsf{c}$ & $\jmath_i^\mathsf{q}$ & $\Delta_i$ \\
    \hline
    $1$ & $6$ & $6$ & $0.5$ & $0.01$ \\
    $2$ & $8$ & $8$ & $0.5$ & $0.01$ \\
    \hline
    \hline
  \end{tabular}
  \caption{Parameters of the example.}
  \label{tab:exampleParam}
\end{table}
From~(\ref{eq:mbarDef}) and~(\ref{eq:residual}), it follows that
$\mbar{1} = \mbar{2} = 2$, and $\rho_1=\frac{5}{6}$,
$\rho_2=\frac{1}{8}$, implying that both functions must always keep
two machines on, and then switch a third one on/off periodically.

From~(\ref{eq:setABCDXlin}), the parameters needed to formulate the
optimization problem of~(\ref{eq:constrSumLeqC}) are: $a_1 = 8.5$,
$a_2 = 8.5$, $b_1 = 2.94\times 10^{-3}$, $b_2 = 4.12\times 10^{-3}$,
and $c=0.02$. Also, from~(\ref{eq:costLowB}) the cost lower bound is
$J^\mathsf{lb}=34$.

The unconstrained solution of~(\ref{eq:solUncontr}) is then given by
$x_1 
= 18.6\times 10^{-3}$ and $x_2 
= 22.0\times 10^{-3}$.
Such a solution, however, violates E2E deadline constraint since
\[
x_1+x_2 = 40.6 \times 10^{-3} > c = 0.02.
\]
Therefore, the constrained solution must be explored.

When solving the constrained solution, the Lagrange multiplier
$\lambda = 26.6$ is the solution of (\ref{eq:lambdaConstr}).
From~(\ref{eq:solConstr}), this gives the solution of
$x_1 = 9.16\times 10^{-3}$, and $x_2 = 10.8\times 10^{-3}$, resulting
in the periods $T_1 = 186.8\times 10^{-3}$ and
$T_2 = 210.6\times 10^{-3}$. Note that the off-time for the two
functions are $T_1^\off = 31.1\times 10^{-3}$ and
$T_2^\off = 184.2\times 10^{-3}$, which are both larger than
$\Delta_i = 0.01$. Note that the E2E-delay for this solution is
exactly the E2E-deadline. Finally, from Eq.~(\ref{eq:maxQi}) we find
that the maximum queue-sizes for this solution are
$\qmax{1} = 155.7\times 10^{-3}$ and $\qmax{2} = 184.3\times 10^{-3}$.
Finally, from~(\ref{eq:costConstr}) the cost for the solution is
$J'' = 34.871$. It should be noted that this example is meant to
illustrate how one can use the design methodology of this section in
order to find the periods $T_1$ and $T_2$ as well as the maximum
queue-sizes $\qmax{1}$ and $\qmax{2}$. In a real setting the incoming
traffic will likely be around million requests per second,
\cite{Zhang:aa}.

\section{Design of  machine-switching period}
\label{sec:periodic}

In the previous section, the service functions were approximated by a
linear lower-bound, which allowed us to find a period $T_i$ for each
function. However, such an approximation leads to an extra cost. In
this section, the exact expression of the service functions will be
considered. Since the exactness of the service functions leads to an
increases in the complexity, the design problem
of~\eqref{eq:optProblem} will be solved while letting every function
switch its additional machine on/off with the same period, $T_i=T$.

The common period $T$ of the schedule, by which every function
switches its additional machine on/off, is the only design variable in
the optimization problem~\eqref{eq:optProblem}. As proved later in
Lemma~\ref{lem:qimax} and Lemma~\ref{lem:E2Edelay}, the maximum queue
size $\qmax{i}$ of any function $F_i$ and the E2E delay
$\maxdelay{1}{n}$ are both proportional to the switching period $T$.
The intuition behind this fact is that the longer the period $T$ is,
the longer a function will have to wait with the additional machine
being off, before turning it on again.  During this interval of time,
each function is accumulating work and consequently both the maximum
queue size and the delay grows with $T$.

With these hypothesis, the cost function of the optimization
problem~(\ref{eq:optProblem}) becomes
\begin{equation}
  \label{eq:costPeriodFunT}
  J(T) = aT+\sum_{i:T<\overline{T}_i}\jmath_i^\mathsf{c}(1-\rho_i)+
  \sum_{i:T\geq \overline{T}_i}\jmath_i^\mathsf{c}\frac{\Delta_i}{T}+
  J^\mathsf{lb},
\end{equation}
where $J^\mathsf{lb}$ is the lower bound given by~(\ref{eq:costLowB})
and $a = \sum_{i=1}^nj_i^\mathsf{q} \alpha_i$, where $\alpha_i$ is
given by Lemma~\ref{lem:qimax}. Furthermore, $\overline{T}_i$ (defined
in~(\ref{eq:condToffSmall})) represents the value of the period below
which it is not feasible to switch the additional machine off and then
on again ($T<\overline{T}_i\ \Leftrightarrow\ T_i^\off<\Delta_i$). In
fact, $\forall i$ with $T<\overline{T}_i$ we pay the full cost of
having $\mbar{i}+1$ machines always on.

The deadline constraint in~(\ref{eq:optProblem}), can be simply
written as
\[
  T\leq c := \frac{D^\mathsf{max}}{\sum_{i=1}^n\delta_i},
\]
with $\delta_i$ opportune constants, given in
Lemma~\ref{lem:E2Edelay}.

The cost $J(T)$ of~(\ref{eq:costPeriodFunT}) is a continuous function
of one variable $T$. It has to be minimized over the closed interval
$[0,c]$. Hence, by the Weiersta\ss's extreme-value theorem, it has a
minimum. To find this minimum, we just check all (finite) points at
which the cost is not differentiable and the ones where the derivative
is equal to zero. Let us define all points in $[0,c]$ in which $J(T)$
is not differentiable:
\begin{equation}
  \mathcal{C}=\{\overline{T}_i:\overline{T}_i<c\} \cup \{0\} \cup \{c\}.
  \label{eq:defC}
\end{equation}
We denote by $p=|\mathcal{C}|\leq n+2$ the number of points in
$\mathcal{C}$. Also, we denote by $c_k\in\mathcal{C}$ the points in
$\mathcal{C}$ and we assume they are ordered increasingly
$c_1<c_2<\ldots<c_p$.  Since the cost $J(T)$ is differentiable over
the open interval $(c_k,c_{k+1})$, the minimum may also occur at an
interior point of $(c_k,c_{k+1})$ with derivative equal to zero. Let
us denote by $\mathcal{C}^*$ the set of all interior points of
$(c_k,c_{k+1})$ with derivative of $J(T)$ equal to zero, that is
\begin{equation}
  \mathcal{C}^* = \{c_k^*: k=1,\ldots,p-1,\ c_k< c_k^*< c_{k+1}\}
  \label{eq:defCstar}
\end{equation}
with
\[
c_k^* = \sqrt{\frac{\sum_{i:\overline{T}_i<c_{k+1}}\jmath_i^\mathsf{c}\Delta_i}{a}}.
\]
Then, the optimal period is given by
\begin{equation}
  \label{eq:optPeriodSameT}
  T^* = \arg\min_{T\in\mathcal{C}\cup\mathcal{C}^*}\{J(T)\}.
\end{equation}

Next, we illustrate an example of solution of the design problem.
Later, Lemma~\ref{lem:qimax} and Lemma~\ref{lem:E2Edelay} provide the
expression of maximum queue size $\qmax{i}$ and the E2E delay
$\maxdelay{1}{n}$, as function of the switching period $T$.

\paragraph{Example}
As in Section~\ref{sec:linear}, we use an example to illustrate the
solution of the optimization problem of a service chain containing two
functions. The input to the service-chain has a rate of $r_0(t)=r=17$.
Every request has an E2E-deadline of $D^\mathsf{max} = 0.02$. The
parameters of the two functions are reported in
Table~\ref{tab:exampleParam}.

The input $r_0(t) = r$ can be seen as dummy function $F_0$ preceding
$F_1$, with $\ns{0} = r$, $\mbar{0}=1$, and $\rho_0 = 0$ (from
Equations~\eqref{eq:mbarDef}--\eqref{eq:residual}). Also, as in the
example of Section~\ref{sec:linear}, $\mbar{1} = \mbar{2} = 2$,
$\rho_1 = 0.833$, and $\rho_2 = 0.125$. This in turn leads to
$\overline{T}_1 = 60.0\times10^{-3}$ and
$\overline{T}_2 = 11.4\times10^{-3}$, where $\overline{T}_i$ is the
threshold period for function $F_i$, as defined
in~\eqref{eq:condToffSmall}. From Lemma~\ref{lem:qimax} it follows
that the parameter $a$ of the cost function \eqref{eq:costPeriodFunT}
is $a =
0.792$,
while from Lemma~\ref{lem:E2Edelay} the parameters $\delta_i$
determining the queuing delay introduced by each function, are
$\delta_1 = 49.0\times10^{-3}$ and $\delta_2 = 22.1\times10^{-3}$,
which in turn leads to
\[
  c = \frac{D^{\max}}{\delta_1 + \delta_2} 
= \frac{0.02}{71.1\times 10^{-3}} = 281\times 10^{-3}.
\]

Since $\overline{T}_2 < \overline{T}_1 < c$, the set $\mathcal{C}$
of~(\ref{eq:defC}) containing the boundary is
\[
\mathcal{C}=\{0,\underbrace{0.00114}_{\overline{T}_2},\underbrace{0.060}_{\overline{T}_1}, \underbrace{0.281}_c\}.
\]
To compute the set $\mathcal{C}^*$ of interior points with derivative
equal to zero defined in (\ref{eq:defCstar}), which is needed to
compute the period with minimum cost from~(\ref{eq:optPeriodSameT}),
we must check all intervals with boundaries at two consecutive points
in $\mathcal{C}$. In the interval $(0,\overline{T}_2)$ the derivative
of $J$ is never zero. When checking the interval $(\overline{T}_2,
\overline{T}_1)$, the derivative is zero at
\[
  c_1^* = \sqrt{\frac{\jmath_2^\mathsf{c}\Delta_2}{a}} = 0.318,
\]
which, however, falls outside the interval. Finally, when checking the
interval $(\overline{T}_1, c)$ the derivative is zero at
\[
c_2^* = \sqrt{\frac{\jmath_1^\mathsf{c}\Delta_1 + \jmath_2^\mathsf{c}\Delta_2}{a}} = 0.421>c=0.281.
\]
Hence, the set of points with derivative equal to zero is
$\mathcal{C}^*=\emptyset$. By inspecting the cost at points in
$\mathcal{C}$ we find that the minimum occurs at $T^*=c=0.281$, with
cost $J(T^*)=34.7$.  It should noted that this solution provides a
lower cost than the one found by the linear approximation (in
Section~\ref{sec:linear}), that is $34.871$. This, however, is not
true in general.

To conclude the example we show in
Figure~\ref{fig:example-periodic-state-space} the state-space
trajectory for the two queues. There one can see how the two queues
grows and shrinks depending on which of the two functions has their
additional machine on. Again, it should be noted that this example is
meant to illustrate how one can use the design methodology of this
section in order to find the best period $T$. In a real setting the
incoming traffic will likely be around million requests per second,
\cite{Zhang:aa}.
\begin{figure}[ht]
  \centering
  \includegraphics[width=\columnwidth]{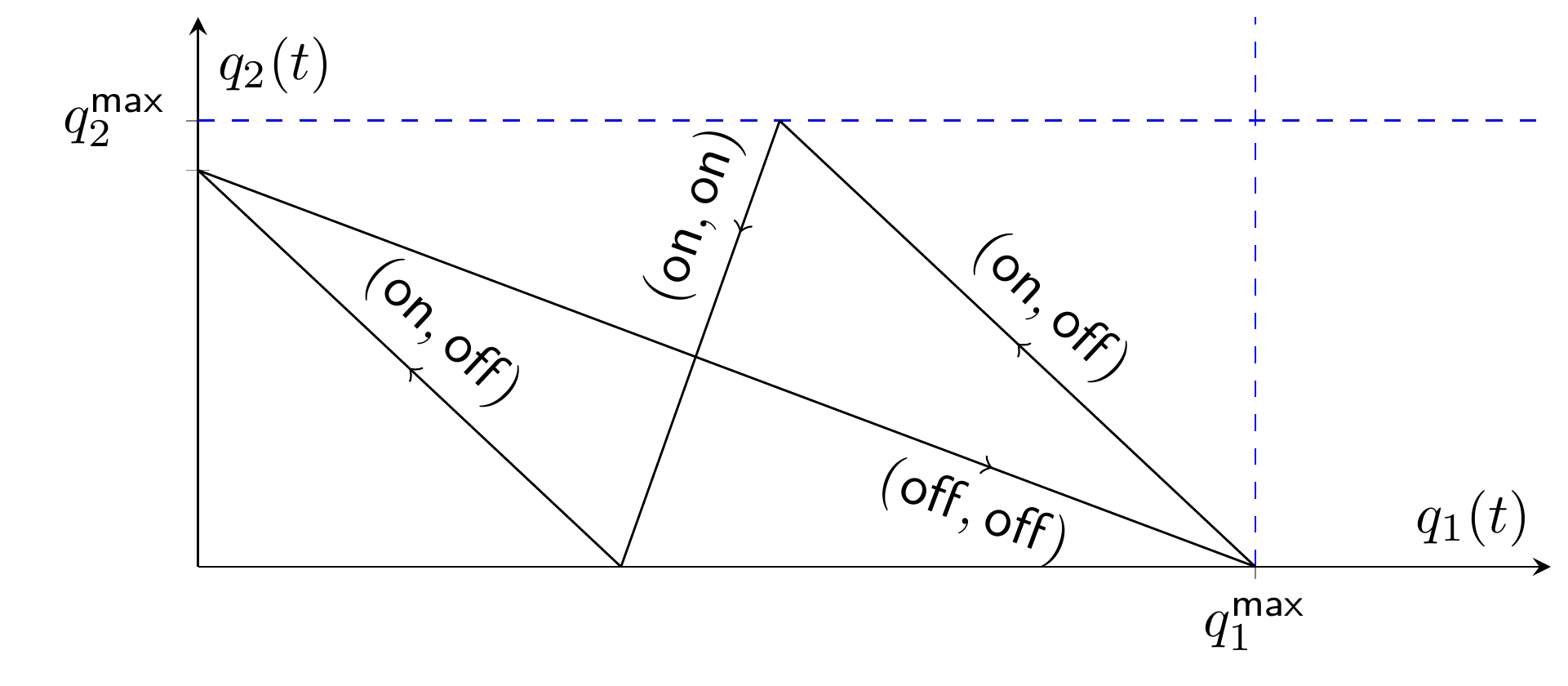}
  \caption{\label{fig:example-periodic-state-space} State-space
    trajectory for the example in
    Section~\ref{sec:periodic}. $(\on,\,\off)$ correspond to $F_1$
    having its additional machine on, while $F_2$ has its extra
    machine off.}
\end{figure}

Next we derive the expression of the maximum queue size $\qmax{i}$ as
function of the switching period $T$.
\begin{lemma}
  The maximum queue size $\qmax{i}$ at function $F_i$ is
  \begin{equation}
    \label{eq:qMaxPeriodic}
    \qmax{i} = T \times \alpha_i,
  \end{equation}
  where
  \begin{multline*}
    \alpha_i = \max \big\{
    \rho_i\big(\ns{i}(1-\rho_i) - \ns{i-1}(1-\rho_{i-1})\big),\\
    (1-\rho_{i-1})(\ns{i-1}\rho_{i-1} - \ns{i}\rho_{i}),\\
    \rho_{i-1}\big(\ns{i-1}(1-\rho_{i-1}) - \ns{i}(1-\rho_i)\big),\\
    (1-\rho_i)(\ns{i}\rho_{i} - \ns{i-1}\rho_{i-1})
    \big\},
  \end{multline*}
  with $\rho_i$ as defined in~(\ref{eq:residual}), and $T$ being the
  period of the switching scheme, common to all functions.
  \label{lem:qimax}
\end{lemma}

\begin{proof}
The queue size over time $q_i(t)$ is a continuous, piecewise-linear
function, since both the input and the service rates are piecewise
constant, and the queue size is defined by
Eq.~\eqref{eq:def-queue}. Hence, if at $t^*$ the function $q_i(t)$
takes its maximum value, it must necessarily happen that $\partial
q_i(t)/\partial t\geq 0$ in a left-neighbourhood of $t^*$ and $\partial
q_i(t)/\partial t\leq 0$ in a right-neighbourhood of $t^*$.

To find the value of $\partial q_i(t) /\partial t$, one needs to
distinguish among the four possible cases, Case~(1a), Case~(1b),
Case~(2a), and Case~(2b), depending on the nominal speeds $\ns{i-1}$
and $\ns{i}$, as is shown in Table~\ref{tab:cases}. These cases, in
turn, determine the sign of $\partial q_i(t)/\partial t$, as
summarised in Table~\ref{tab:vectors}. Note that for $F_i = F_1$, one
should consider the input as $F_{i-1} = F_0$ with $\ns{0} = r$,
leading to $\mbar{0} = 1$ and $\rho_0 = 0$, which would then belong to
Case~(2b).

\begin{table}[ht]
  \centering
  \begin{tabular}[t]{c|c|c}
    \hline
    \hline
    Case~(1a) & $(\mbar{i}+1)\ns{i} \geq (\mbar{i-1}+1)\ns{i-1}$ &
    \multirow{2}{*}{$\mbar{i}\ns{i} \geq \mbar{i-1}\ns{i-1}$} \\
    Case~(1b) & $(\mbar{i}+1)\ns{i} < (\mbar{i-1}+1)\ns{i-1}$ & \\ 
    \hline
    Case~(2a) & $(\mbar{i}+1)\ns{i} \geq (\mbar{i-1}+1)\ns{i-1}$ &
    \multirow{2}{*}{$\mbar{i}\ns{i} < \mbar{i-1}\ns{i-1}$} \\
    Case~(2b) & $(\mbar{i}+1)\ns{i} < (\mbar{i-1}+1)\ns{i-1}$ & \\ 
    \hline
    \hline
  \end{tabular}
  \caption{The four possible cases that one needs to distinguish
    among. Each case is a function of the nominal speeds $\ns{i}$ and
    $\ns{i-1}$.}
  \label{tab:cases}
\end{table}

\begin{table}[ht]
  \centering
  \begin{tabular}[t]{c|cccc}
    \hline
    \hline
    $m_{i-1}(t)$ & $\mbar{i-1}$ & $\mbar{i-1}+1$ & $\mbar{i-1}+1$ & $\mbar{i-1}$ \\
    $m_i(t)$ & $\mbar{i}$ & $\mbar{i}$ & $\mbar{i}+1$ & $\mbar{i}+1$ \\
    \hline
    Case (1a) & $\leq 0$ & $\geq 0$ & $\leq 0$ & $\leq 0$ \\
    Case (1b) & $\leq 0$ & $\geq 0$ & $> 0$ & $\leq 0$ \\
    Case (2a) & $ > 0$ & $\geq 0$ & $\leq 0$ & $\leq 0$ \\
    Case (2b) & $ > 0$ & $\geq 0$ & $> 0$ & $\leq 0$ \\
    \hline
    \hline
  \end{tabular}
  \caption{Sign of $\partial q_i(t)/ \partial t$ as function of the
    number of on-machines within $F_{i-1}$ and $F_i$.}
  \label{tab:vectors}
\end{table}

Next, the maximum queue-size $\qmax{i}$ will be derived for each case.
We will also derive the best time for each function to start its
additional machine, i.e. $\ton{i}$.

\paragraph{Case (1a)}
For this case, illustrated in Figure~\ref{fig:case1a}, the sign of
$\partial q_i(t)/\partial t$ shown in Table~\ref{tab:vectors}, implies
that $q_i(t)$ grows only when $m_i(t)=\mbar{i}$ and
$m_{i-1}(t) = \mbar{i-1}+1$.  From this condition, the $i$-th queue
can start to decrease either when $m_i(t)\to\mbar{i}+1$ or
$m_{i-1}(t)\to\mbar{i-1}$. In the first case, the rate of decrease is
\begin{align*}
  -\partial q_i(t)/\partial t &= \big((\mbar{i}+1)\ns{i} - (\mbar{i-1}+1)\ns{i-1}\big)\\
  &= \big(\ns{i}(1-\rho_i) - \ns{i-1}(1-\rho_{i-1})\big),
\end{align*}
and such a state lasts for $T_i^\on$ (during the interval of length
$T_i^\on$ in Figure~\ref{fig:case1a}). This therefore yields a local
maximum of:
\begin{align}
  \label{eq:1a-qon}
  q_i(\ton{i}) &=  
                 T \rho_i\big(\ns{i}(1-\rho_i) - \ns{i-1}(1-\rho_{i-1})\big).
\end{align}
It is easy to verify that changing the on-time $\ton{i}$ to instead be
later will yield a larger local maximum, and changing it to instead be
earlier will yield a negative queue size. The given $\ton{i}$ is thus
the optimal one, and can be expressed relative to $\ton{i-1}$ as:
\begin{align}
  \label{eq:1a-ton}
  \ton{i} = \ton{i-1} + T\rho_i\frac{\ns{i}(1-\rho_i) - \ns{i-1}(1-\rho_{i-1})}{\ns{i-1}(1-\rho_{i-1}) + \ns{i}\rho_i}
\end{align}

On the other hand, the local maximum when $m_{i-1}(t)\to\mbar{i-1}$ is
determined by the interval of length $T_{i-1}^\off$, as shown in
Figure~\ref{fig:case1a}, that is
\[
T_{i-1}^\off(\mbar{i}\ns{i} - \mbar{i-1}\ns{i-1}) 
= T(1-\rho_{i-1})(\ns{i-1}\rho_{i-1} - \ns{i}\rho_{i}).
\]
By taking the maximum of the two local maxima, we find
\begin{multline*}
  \qmax{i} = T \max \big\{\rho_i\big(\ns{i}(1-\rho_i) - \ns{i-1}(1-\rho_{i-1})\big) \,,  \\
  (1-\rho_{i-1})(\ns{i-1}\rho_{i-1} - \ns{i}\rho_{i}) \big\}.
\end{multline*}
\begin{figure}[ht]
  \centering
  \includegraphics[scale=1]{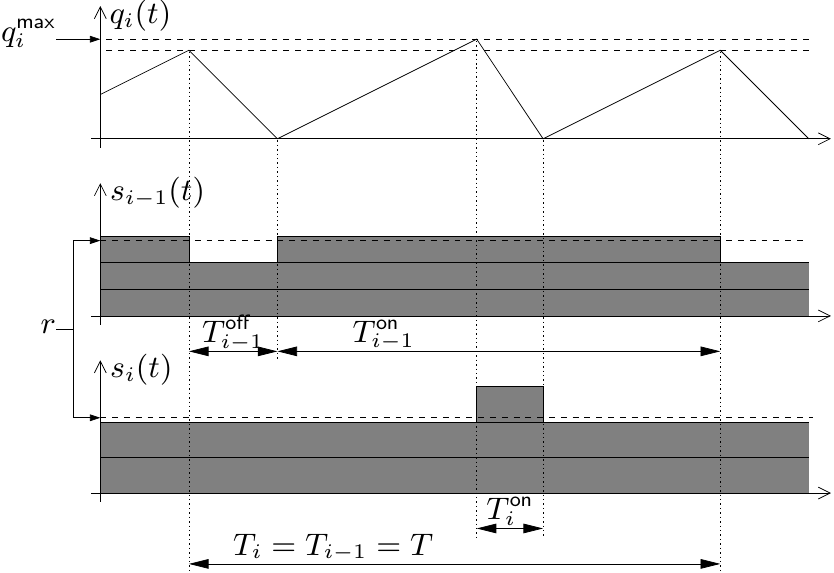}
  \caption{Case (1a): service schedule and queue $q_i(t)$. 
    In this
    example: $r=17$, $\ns{i-1}=6$, $\ns{i}=8$, $T=120$,
    $T_{i-1}^\on=100$, $T_i^\on=15$, $\qmax{i}=90$.}
  \label{fig:case1a}
\end{figure}

\paragraph{Case (1b)}
As shown in Table~\ref{tab:vectors}, the queue size $q_i(t)$ grows if
and only if $\mbar{i-1}+1$ machines are running within function
$F_{i-1}$. The maximum queue size, then, is attained at the instant
when such a machine is switched off. To analyse this case, we
distinguish between two cases: $T_i^\on \geq T_{i-1}^\on$ (illustrated
in Figure~\ref{fig:case1b1}) and $T_i^\on < T_{i-1}^\on$
(Figure~\ref{fig:case1b2}). In both cases, to minimize $\qmax{i}$, the
function $F_i$ must start the extra machine simultaneously as
$F_{i-1}$ start its additional machine in order to reduce the rate of
growth of the $i$-th queue, i.e.
\begin{align}
  \label{eq:1b-ton}
  \ton{i} = \ton{i-1}.
\end{align}
 Note that the queue size for function $F_i$
will therefore be zero when it switches on the additional machine,
\begin{equation*}
  q_i(\ton{i}) = 0.  
\end{equation*}

To compute $\qmax{i}$, we examine both when
$T_i^\on \geq T_{i-1}^\off$ (illustrated in Figure~\ref{fig:case1b1}),
as well as when $T_i^\on < T_{i-1}^\off$ (illustrated in
Figure~\ref{fig:case1b2}). By considering them both together, we find
\begin{multline*}
  \qmax{i} = \max\big\{T_{i-1}^\on\big(\ns{i-1}(1-\rho_{i-1}) - \ns{i}(1-\rho_i)\big),\\
  T_{i-1}^\off  (\ns{i-1}\rho_{i-1} - \ns{i}\rho_{i})
  \big\} 
\end{multline*}
and, by considering the expressions of $T_{i-1}^\on$ and
$T_{i-1}^\off$ of Eq.~(\ref{eq:TiON}) it can be written as:
\begin{multline*}
  \qmax{i} = T\max\big\{\rho_{i-1}\big(\ns{i-1}(1-\rho_{i-1}) - \ns{i}(1-\rho_i)\big),\\
  (1-\rho_{i-1})( \ns{i-1}\rho_{i-1} - \ns{i}\rho_{i})
  \big\}.
\end{multline*}

\begin{figure}[htb]
  \centering
  \includegraphics[scale=1]{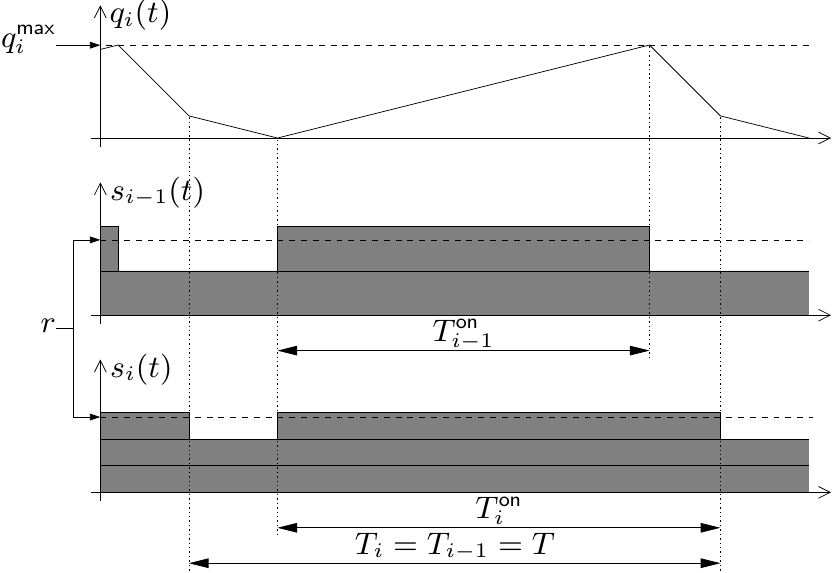}
  \caption{Case (1b), $T_i^\on\geq T_{i-1}^\on$. 
    In this
    example: $r=17$, $\ns{i-1}=10$, $\ns{i}=6$, $T=120$,
    $T_{i-1}^\on=84$, $T_i^\on=100$, $\qmax{i}=168$.}
  \label{fig:case1b1}
\end{figure}

\begin{figure}[htb]
  \centering
  \includegraphics[scale=1]{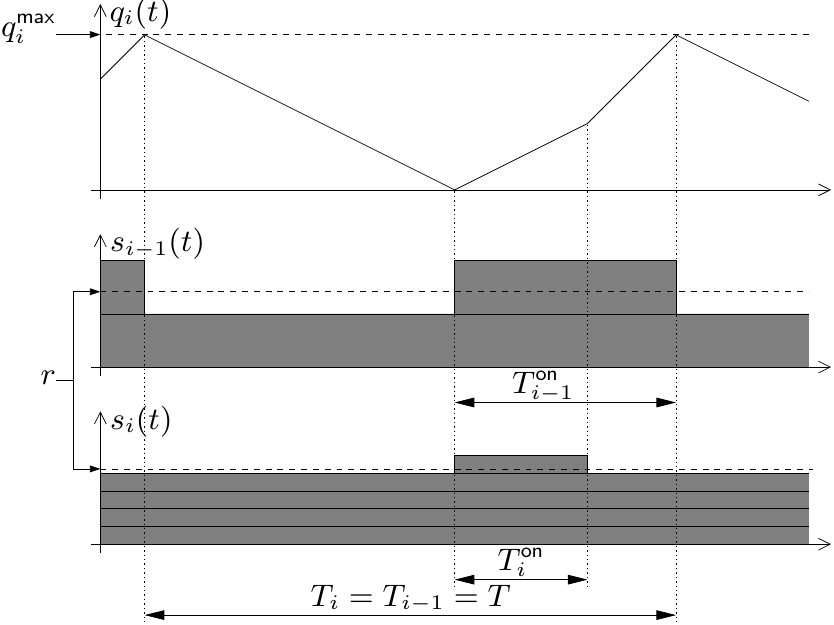}
  \caption{Case (1b), $T_i^\on< T_{i-1}^\on$. 
    In this
    example: $r=17$, $\ns{i-1}=12$, $\ns{i}=4$, $T=120$,
    $T_{i-1}^\on=50$, $T_i^\on=30$, $\qmax{i}=280$.}
  \label{fig:case1b2}
\end{figure}

\paragraph{Case (2a)}
This case is essentially the same as Case~(1b). As shown by
Table~\ref{tab:vectors}, the only difference is that $q_i(t)$ is
reduced whenever $F_i$ has its extra machine on, and grows whenever it
is off. This then implies that the maximum queue size is attained when
$F_i$ switches on the extra machine. To minimize $\qmax{i}$, the queue
size of $F_i$ should therefore be such that the queue is empty when it
switches off the additional machine. Note that this corresponds to
both $F_i$ and $F_{i-1}$ switching off their additional machine
simultaneously (compare with Case~(1b) where the two functions
switches on their additional machine simultaneously). The time when
$F_i$ should switch on its additional machine is thus:
\begin{align}
  \label{eq:2a-ton}
  \ton{i} = \underbrace{\ton{i-1} + T_{i-1}^\on}_{\toff{i-1} = \toff{i}} - T_i^\on = \ton{i-1} + T(\rho_{i-1} - \rho_i).
\end{align}
Note that for this case have to consider both
$T_i^\on \geq T_{i-1}^\on$ and $T_i^\on < T_{i-1}^\on$ when computing
$q_i(\ton{i})$:
\begin{equation}
  \label{eq:2a-qon}
  q_i(\ton{i})\! =\! 
  \begin{cases}
    T(1-\rho_i)(\ns{i}\rho_{i} - \ns{i-1}\rho_{i-1}),\! &T_i^\on\! \geq\! T_{i-1}^\on \\
    T\rho_i\bigl( \ns{i}(1-\rho_i) - \ns{i-1}(1-\rho_{i-1})),\! &T_i^\on\! <\! T_{i-1}^\on
  \end{cases}
\end{equation}
The maximum queue size is, as stated earlier, found when $F_i$
switches on its extra machine. By considering
$T_i^\on \geq T_{i-1}^\on$ and $T_i^\on < T_{i-1}^\on$ together, the
expression for $\qmax{i}$ can be combined into:
\begin{multline*}
  \qmax{i} = T\max\big\{
  \rho_i\big(\ns{i}(1-\rho_i) - \ns{i-1}(1-\rho_{i-1})\big),\\
  (1-\rho_i)(\ns{i}\rho_{i} - \ns{i-1}\rho_{i-1})
  \big\} 
\end{multline*}

\paragraph{Case (2b)}
Table~\ref{tab:vectors} show the similarity between this case and
Case~(1a), with the difference being that for this case $q_i(t)$ only
shrinks when $m_i(t) = \mbar{i}+1$ and $m_{i-1}(t) = \mbar{i-1}$.
Therefore, $q_i(t)$ will always grow when $m_{i-1}(t) = \mbar{i-1}+1$.
To reduce the rate of this growth $F_i$ should therefore have its
extra machine on whenever $F_{i-1}$ has its extra machine
on. Furthermore, to reduce the local maximum attained at the end of
this growth, $F_i$ should switch on its additional machine such that
$q_i(t)$ is empty at the start of it, i.e.  $q_i(\ton{i-1}) =
0$.
Furthermore, since $q_i(t)$ grows when both $F_i$ and $F_{i-1}$ has
its additional machine off, there is also a local maximum for $q_i(t)$
attained when $F_i$ switches on its additional machine. To minimize
this local maximum, $F_i$ should ensure that $q_i(t)$ is empty when it
switches off its additional machine, i.e. $q_i(\toff{i}) = 0$. The
on-switching time should thus be:
\begin{align}
  \label{eq:2b-ton}
  \ton{i} &= \ton{i-1} - \frac{q_i(\ton{i})}{\ns{i}(1-\rho_i) + \ns{i-1}\rho_{i-1}}
\end{align}
where
\begin{align}
  q_i(\ton{i}) &= T_i^\off (\mbar{i-1}\ns{i-1} - \mbar{i}\ns{i}) \nonumber \\
  \label{eq:2b-qon}
  &= T(1-\rho_i) (\ns{i}\rho_{i} + \ns{i-1}\rho_{i-1}).
\end{align}
The other local maximum, occurring when $F_{i-1}$ switches off its
additional machine is therefore:
\begin{align*}
    q_i(\toff{i-1}) &= T_{i-1}^\on\big((\mbar{i-1}+1)\ns{i-1} - (\mbar{i}+1)\ns{i}\big) \\
                    &= T\rho_{i-1}\big(\ns{i-1}(1-\rho_{i-1}) - \ns{i}(1-\rho_i)\big)
\end{align*}

The maximum queue-size for this case thus given by:
\begin{multline*}
  \qmax{i} = T \max \big\{
  \rho_{i-1}\big(\ns{i-1}(1-\rho_{i-1}) - \ns{i}(1-\rho_i)\big), \\
  (1-\rho_i)(\ns{i}\rho_{i} - \ns{i-1}\rho_{i-1}) 
  \big\}. 
\end{multline*}

\paragraph{Conclusion}
By taking the maximum among all four cases,
Equation~(\ref{eq:qMaxPeriodic}) is found and the Lemma is proved.
\end{proof}

The expression of $\qmax{i}$ of Eq.~(\ref{eq:qMaxPeriodic}) suggests a
property that is condensed in the next Corollary.
\begin{corollary}
  The maximum queue qize $\qmax{i}$ at any function $F_i$ is bounded,
  regardless of the rate $r$ of the input.
  \label{cor:boundQ}
\end{corollary}
\begin{proof}
  From the definition of $\rho_i$ of Eq.~\eqref{eq:residual}, it
  always holds that $\rho_i \in [0,1)$.  Hence, from the expression
  of~(\ref{eq:qMaxPeriodic}), it follows that $\qmax{i}$ is always
  bounded.
\end{proof}

The second ingredient needed to solve the optimal design problem is
the expression of the end-to-end delay.

\begin{lemma}
  With a constant input rate, $r_0(t) = r$, the longest end-to-end
  delay $\maxdelay{i}{n}$ for any request passing through functions
  $F_1$ thru $F_n$ is
  \begin{align}
    \label{eq:periodicMaxDelay}
    \maxdelay{1}{n} = T \times \sum_{i=1}^n \delta_i.
  \end{align}
  with $\delta_i$ being an opportune constant that depends on $r$,
  $\ns{i}$, and $\ns{i-1}$.
  \label{lem:E2Edelay}
\end{lemma}

\begin{proof}
  With a constant input $r_0(t)=r$ to the service chain, the maximum
  E2E delay for function $F_i$ is given by
\[
 \maxdelay{1}{i} = \max_t \frac{R_0(t) - S_i(t)}{r} =
 \max_t \Big(t-\frac{S_i(t)}{r}\Big),
\]
with $S_i(t)$ being the cumulative served request by $F_i$, as in
Eq.~\eqref{eq:cumServReq}, and $R_0(t)$ is the cumulative arrived
requests of~\eqref{eq:intR}.  Since $S_i(t)$ is piecewise linear
function, growing with rates $\ns{i}\mbar{i}$ or $\ns{i}(\mbar{i}+1)$
as illustrated in in Figure~\ref{fig:Dstar}, it follows that the
maximum end-to-end delay up to the $i$-th function, $\maxdelay{1}{i}$,
is attained when $F_i$ switches on the additional machine (denoted by
$t_i^\on$), that is
\begin{equation}
  \maxdelay{1}{i}
  = \max_t \Big(t-\frac{S_i(t)}{r}\Big)
  = t_i^\on-\frac{S_i(t_i^\on)}{r}.
  \label{eq:D1iTon}
\end{equation}

\begin{figure}[t]
  \centering
  \includegraphics[width=\columnwidth]{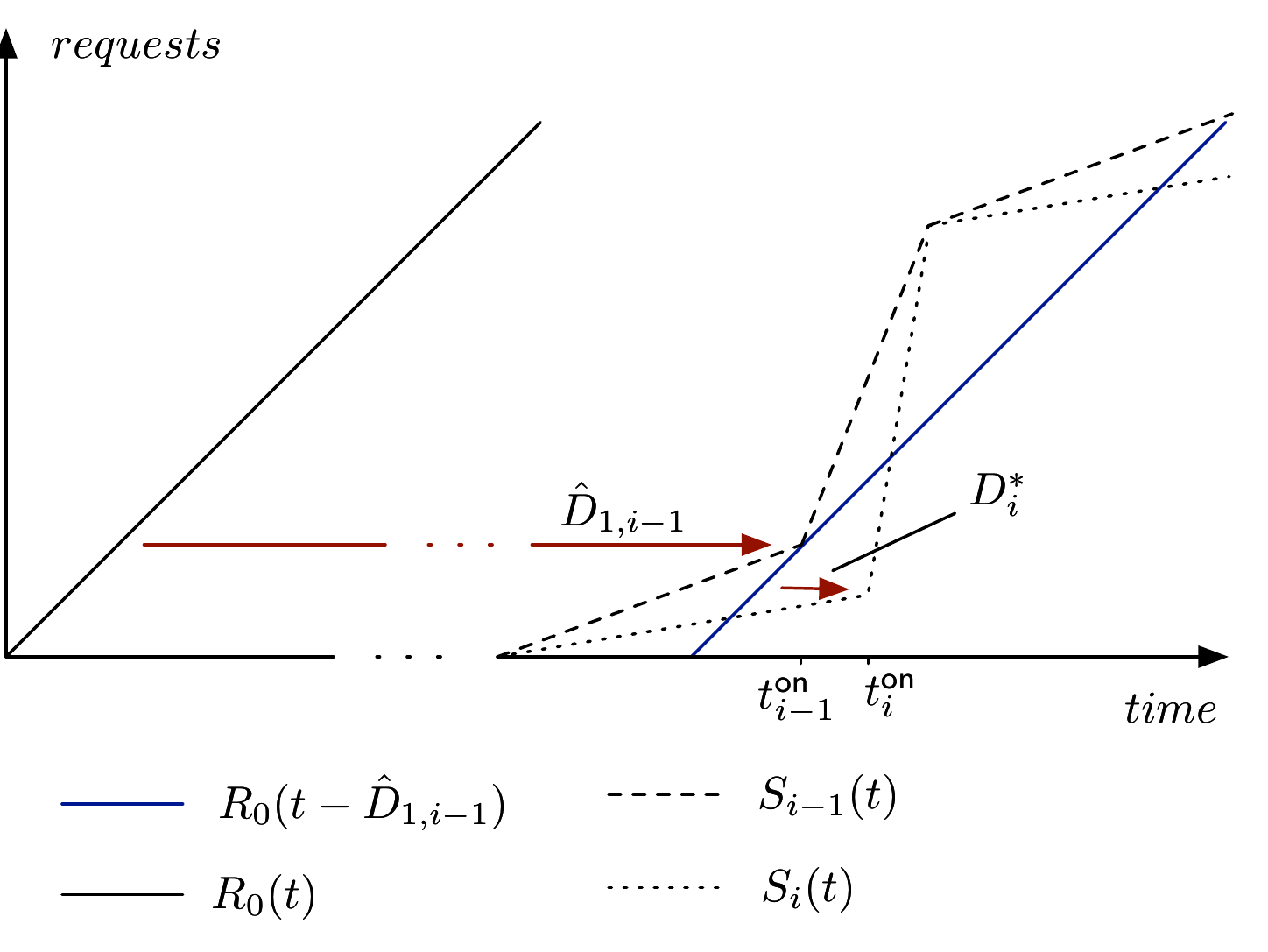}
  \caption{\label{fig:Dstar} Illustration of how function $F_i$ adds
    $D_i^*$ to the maximum E2E delay. $R_0(t)$ is the cumulative
    arrived requests into the service chain, $S_i(t)$ and $S_{i-1}(t)$
    are the cumulative served requests by function $F_i$ and $F_{i-1}$
    respectively. $R_0(t-\maxdelay{1}{i-1})$ is the linear lower-bound
    approximation for $S_{i-1}(t)$. The maximum E2E when adding
    function $F_i$ to the service-chain is then given as
    $\maxdelay{1}{i} = \maxdelay{1}{i-1}+D_i^*$.
  }
\end{figure}

As illustrated in Figure~\ref{fig:Dstar}, function $F_i$ will add
$D_i^*$ to the maximum E2E delay up to function $F_{i-1}$. Therefore,
it is possible to write the maximum E2E delay up to the $i$-th function
as
\begin{equation}
 \label{eq:E2Edelay}
   \maxdelay{1}{i} = \maxdelay{1}{i-1} + D_i^*.
\end{equation}
Equation~\eqref{eq:D1iTon} then implies that

\begin{align*}
  D_i^* &= \maxdelay{1}{i} - \maxdelay{1}{i-1} \\
        &= \ton{i} - \frac{S_i(\ton{i})}{r} - \ton{i-1} + \frac{S_{i-1}(\ton{i-1})}{r} \\
        &= \ton{i} - \ton{i-1} + \frac{S_{i-1}(\ton{i-1}) - S_i(\ton{i})}{r} \\
        &= \ton{i} - \ton{i-1} + \frac{\overbrace{S_{i-1}(\ton{i}) - S_i(\ton{i})}^{q_i(t_i^\on)} + \int_{\ton{i}}^{\ton{i-1}}s_{i-1}(x)dx}{r} \\
        &= \ton{i} - \ton{i-1} + \frac{q_i(t_i^\on)}{r} + \frac{1}{r}\int_{\ton{i}}^{\ton{i-1}}s_{i-1}(x)dx \\
        &= \ton{i} - \ton{i-1} + \frac{q_i(\ton{i})}{r} + (\ton{i-1} - \ton{i})\frac{s_{i-1}^*}{r} \\
        &= \frac{q_i(\ton{i})}{r} + (\ton{i} - \ton{i-1})(1-\frac{s_{i-1}^*}{r}),
\end{align*}
where
$\int_{\ton{i}}^{\ton{i-1}}s_{i-1}(x)dx = (\ton{i-1}-\ton{i})\times
s_{i-1}^*$
since $s_{i-1}(t)$ is a piecewise constant function, changing value
only in $\ton{i-1}$. The values of $s_{i-1}^*$ depend on whether
$F_{i-1}$ has its additional machine on or off during this
time-interval. It should be noted that when
$t_i^\on \geq t_{i-1}^\on$, function $F_{i-1}$ will start its
additional machine before $F_{i}$ does so, and $F_{i-1}$ will
therefore have $(\mbar{i-1}+1)$ machines on during the time-interval
$[t_{i-1}^\on,\, t_i^\on]$. On the other hand, if
$t_i^\on < t_{i-1}^\on$, it follows that $F_i$ will start its
additional machine before $F_{i-1}$ does so, and $F_{i-1}$ will only
have $\mbar{i-1}$ machines on during the time-interval
$[t_{i-1}^\on,\, t_i^\on]$. Hence, $s_{i-1}^*$ can be written as:
\begin{equation*}
    s_{i-1}^* = 
  \begin{cases}
    \ns{i-1}(\mbar{i-1}+1), & \ton{i} \geq \ton{i-1} \\
    \ns{i-1}\mbar{i-1}, & \ton{i} < \ton{i-1}
  \end{cases}.
\end{equation*}
It should also be noted that $\ton{i}$ and $\ton{i-1}$ are such that
the time between them is the smallest possible. Hence,
$(\ton{i}-\ton{i-1})$ might be positive or negative, and corresponds to
the expressions derived in Lemma~\ref{lem:qimax},
Eqs.~\eqref{eq:1a-ton}, \eqref{eq:1b-ton}, \eqref{eq:2a-ton},
and~\eqref{eq:2b-ton} for Case~(1a)--Case~(2b)
respectively. 

When $t_i^\on \geq t_{i-1}^\on$ we can therefore write $D_i^*$ as
\begin{align}
  D_i^* 
  &= (\ton{i} - \ton{i-1})\Big(1-\frac{\ns{i-1}(\mbar{i-1}+1)}{r}\Big) + \frac{q_i(\ton{i})}{r} \nonumber \\
  &= \frac{\ns{i-1}}{r}(\ton{i} - \ton{i-1})\Big(\underbrace{\frac{r}{\ns{i-1}} - \mbar{i}}_{=\rho_{i-1}} - 1\Big) + \frac{q_i(\ton{i})}{r} \nonumber \\
  &= \frac{\ns{i-1}}{r}(\ton{i} - \ton{i-1})(\rho_{i-1} - 1) + \frac{q_i(\ton{i})}{r}.
\label{eq:Distar-1}
\end{align}
For the opposite case, when $t_i^\on < t_{i-1}^\on$ we instead get
\begin{align}
  D_i^* 
  &= (\ton{i} - \ton{i-1})\Big(1-\frac{\ns{i-1}\mbar{i-1}}{r}\Big) + \frac{q_i(\ton{i})}{r} \nonumber \\
  &= \frac{\ns{i-1}}{r}(\ton{i} - \ton{i-1})\Big(\frac{r}{\ns{i-1}} - \mbar{i}\Big) + \frac{q_i(\ton{i})}{r} \nonumber \\
  &= \frac{\ns{i-1}}{r}(\ton{i} - \ton{i-1})\rho_{i-1} + \frac{q_i(\ton{i})}{r}.
    \label{eq:Distar-2}
\end{align}

In Lemma~\ref{lem:qimax}, both $(\ton{i} - \ton{i-1})$ and
$q_i(\ton{i})$ were derived for Case~(1a)--(2b) in
Eqs.~\eqref{eq:1a-qon}--\eqref{eq:2b-qon}.  For each of the four
cases, $D_i^*$ is given by: \newline
\paragraph{Case~(1a)} For this case, it always holds that
$t_i^\on \geq t_{i-1}^\on$. Hence by inserting $q_i(t_i^\on)$ of
Eq.~\eqref{eq:1a-qon} and $(\ton{i} - \ton{i-1})$ of
Eq.~\eqref{eq:1a-ton} into Eq.~\eqref{eq:Distar-1} we can write
$D_i^*$ as
\begin{equation*}
  D_i^* =  T\times\frac{1}{r}
  \ns{i}\rho_i^2 \dfrac{\ns{i}(1-\rho_i) - \ns{i-1}(1-\rho_{i-1})}{\ns{i-1}(1-\rho_{i-1}) + \ns{i}\rho_i}.
\end{equation*}
\paragraph{Case~(1b)} For this case Eq.~\eqref{eq:1b-ton} imply that
$t_i^\on = t_{i-1}^\on$ and that $q_i(t_i^\on)$ is always $0$. Hence,
$D_i^*$ will always be 0, implying that $\delta_i=0$.

\paragraph{Case~(2a)} Here one must distinguish between two cases:
$T_i^\on \geq T_{i-1}^\on$ and $T_i^\on < T_{i-1}^\on$. When
$T_i^\on \geq T_{i-1}^\on$ it always hold that
$t_i^\on \leq t_{i-1}^\on$. Hence, by inserting
$(\ton{i} - \ton{i-1})$ and $q_i(t_i^\on)$ given by
Eqs.~\eqref{eq:2a-ton}--\eqref{eq:2a-qon} into Eq.~\eqref{eq:Distar-2}
we can write $D_i^*$ as
\begin{equation*}
  D_i^* = T\times\frac{1}{r}
  \big( \ns{i-1}\rho_{i-1}(\rho_{i-1}-\rho_i) + (1-\rho_i)(\ns{i}\rho_i - \ns{i-1}\rho_{i-1}  \big).
\end{equation*}
When $T_i^\on < T_{i-1}^\on$, it instead holds that
$t_i^\on \geq t_{i-1}^\on$. Therefore, by inserting
Eqs.~\eqref{eq:2a-ton}--\eqref{eq:2a-qon} into Eq.~\eqref{eq:Distar-1}
we can write $D_i^*$ as
\begin{multline*}
  D_i^* = T\times\frac{1}{r} \Big(\rho_i\big(\ns{i}(1-\rho_i) -
  \ns{i-1}(1-\rho_{i-1}) \big) + \\ \ns{i-1}(\rho_{i-1}-1)(\rho_{i-1}
  - \rho_i)\Big).
\end{multline*}
\paragraph{Case~(2b)} For this case, it always holds that
$t_i^\on \leq t_{i-1}^\on$. Therefore, by inserting
$(\ton{i} - \ton{i-1})$ and $q_i(t_i^\on)$ given by
Eq.~\eqref{eq:2b-ton}--\eqref{eq:2b-qon} into Eq.~\eqref{eq:Distar-2}
we can write $D_i^*$ as
\begin{equation*}
  D_i^* = T\times\frac{1}{r} \ns{i}(1-\rho_i)^2  \dfrac{\ns{i}\rho_i
+ \ns{i-1}\rho_{i-1}}{\ns{i}(1-\rho_i) + \ns{i-1}\rho_{i-1}}.
\end{equation*}

\paragraph{Conclusion}
It therefore follows that for all four cases it is possible to write
$ D_i^* = T\times \delta_i$, with $\delta_i$ being an opportune
constant depending only on $r$, $\ns{i}$, and $\ns{i-1}$. Note that
Eqs.~\eqref{eq:mbarDef}--\eqref{eq:residual} imply that $\rho_i$ (and
$\rho_{i-1}$) depend on $r$ and $\ns{i}$ (and
$\ns{i-1}$). Equation~\eqref{eq:E2Edelay} then implies that the maximum
queueing delay is
$ \maxdelay{1}{n} = \sum_{i=1}^nD_i^* = T\times\sum_{i=1}^n \delta_i$,
and the lemma is proved.
\end{proof}

\section{Disturbances}
\label{sec:disturbances}

Until now, the analysis of Sections~\ref{sec:switch},
\ref{sec:linear}, and~\ref{sec:periodic} addressed the case with a
constant input rate $r_0(t)=r$. However, variations from such an ideal
condition can easily be are modeled by adding disturbances. An impulse
disturbance of mass $d_i$ will affect both the maximum queue-size
$\qmax{i}$ and the on-time $T_i^\on$ needed by the addition machine of
$F_i$ to process the extra work. If we denote by $\qhat{i}$ the
\emph{largest queue-size for a system without disturbances}, then the
maximum queue size that can avoid an overflow is
\begin{equation}
  \label{eq:qmax-with-disturbance}
  \qmax{i} = \qhat{i} + d_i.
\end{equation}
The additional time needed by $F_i$ to process the disturbances is
$d_i/\ns{i}$. The only time that the function can find ``free time''
to process work this extra work, is when it normally would be off,
i.e. during $T_i^\off$ in a period. Depending on how big this
disturbance is, it might need several periods worth of
$T_i^\off\text{-time}$ in order to process the extra work. The total
\on-time needed to handle the disturbance, along with the usual
incoming request is therefore:
\begin{equation}
  \label{eq:TionTilde}
  \tilde{T}_i^\on = T\underbrace{\floor{\frac{d_i/\ns{i}}{T_i^\off}}}_{\makebox[0pt][c]{\scriptsize \text{number of full periods needed}}}
  + T_i^\on + T_i^\off\overbrace{\big(\frac{d_i/\ns{i}}{T_i^\off} - \floor{\frac{d_i/\ns{i}}{T_i^\off}}\big)}^{\makebox[0pt][c]{\scriptsize \text{fraction of final $T_i^\off$ needed $\in[0,1)$}}}
\end{equation}
It will therefore take $\floor{(d_i/\ns{i})/T_i^\off} + 1 $
periods before the extra work is processed and the schedule can return
to normal. This assumes that the function does not get any extra
disturbances while processing the first one. Note that
$\tilde{T}_i^\on \rightarrow \infty$ as $T_i^\off \rightarrow 0$,
therefore, should $\tilde{T}_i^\on$ grow very large it would be
necessary to switch on yet another machine. If such a thing would
happen in the $i$-th function, it would thus need to switch between
using $\mbar{i}+1$ and $\mbar{i}+2$ machines, which is the problem
studied in this technical report.

The extra on-time needed changes the desired stop-time for the
additional machine, and if $F_i$ is switching between $\mbar{i}$ and
$\mbar{i}+1$ machines, this would be computed as:
\begin{equation}
  \label{eq:toff-with-distubance}
  t_i^\off = t_i^\on + \tilde{T}_i^\on,
\end{equation}
where $\tilde{T}_i^\on$ is given by \eqref{eq:TionTilde}. Note that
this assumes that $d_i$ is known. Note that it could be measured
indirectly by taking the difference of the \emph{real queue-size} when
switching on the extra machine, denoted by $\tilde{q}_i(t_i^\on)$, and
the \emph{expected queue size} $q_i(t_i^\on)$:
\[
  d_i = \tilde{q}_i(t_i^\on) - q_i(t_i^\on).
\]
Note that Eqs.~\eqref{eq:qmax-with-disturbance} and
\eqref{eq:toff-with-distubance} imply that handling a disturbance will
yield a cost increase due to the extra on-time needed and due to the
extra queue-size needed. However, it will not affect the solution of
the optimization problem, since this added cost is constant and does
not depend on the variable of the optimization problem. In
Figure~\ref{fig:disturbance} we illustrate how the modeling errors can
be modeled as a disturbance and how one can stay on for a longer time
in order to process the extra load and ``catch up''.

\begin{figure}[ht]
  \centering
  \includegraphics[width=\columnwidth]{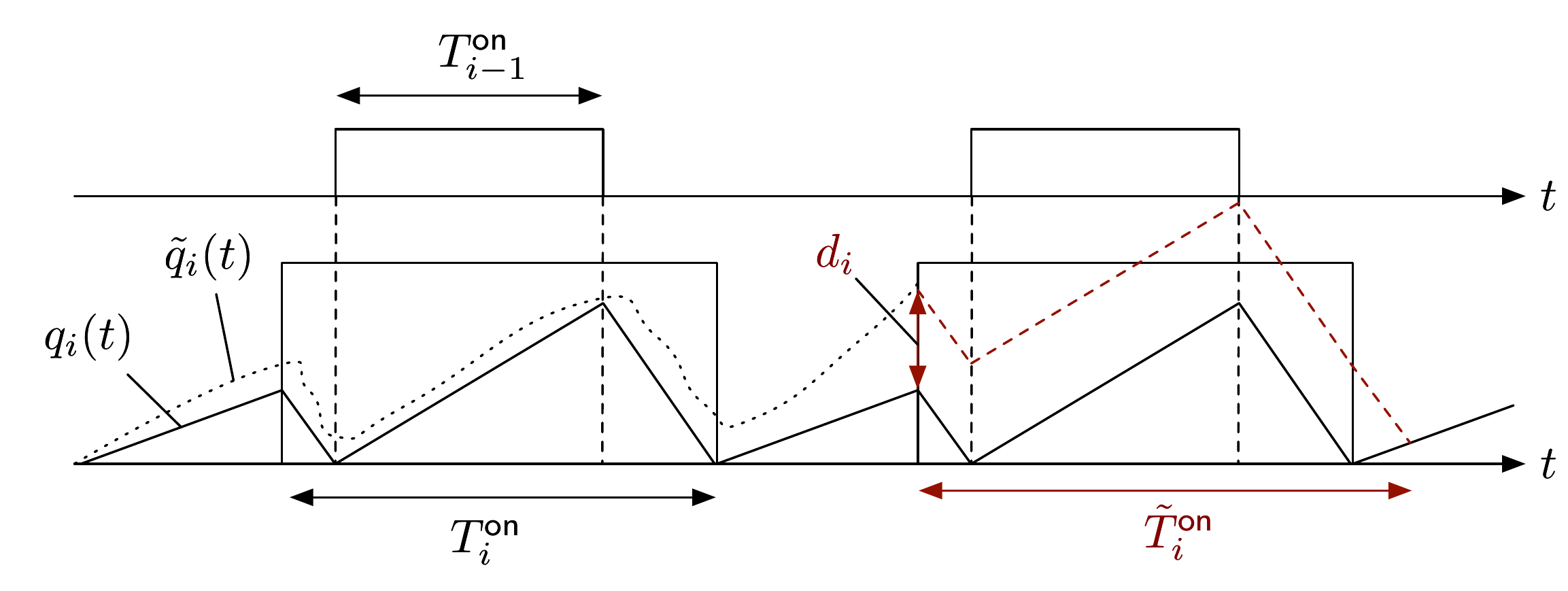}
  \caption{\label{fig:disturbance} Illustration of how a disturbance
    can capture modeling errors and how one can process the
    disturbance in order to ``catch up'' with the model.}
\end{figure}

\section{Summary}
\label{sec:summary}

In this technical report we have developed a general mathematical
model for a service-chain residing in a Cloud environment. This model
includes an input model, a service model, and a cost model. The
input-model defines the input-stream of requests to each NFV along
with end-to-end deadlines for the requests, meaning that they have to
pass through the service-chain before this deadline.  In the
service-model, we define an abstract model of a NFV, in which requests
are processed by a number of machines inside the service function. It
is assumed that each function can change the number of machines that
are up and running, but doing so is assumed to take some time. The
cost-model defines the cost for allocating compute- and storage
capacity, and naturally leads to the optimization problem of how to
allocate the resources. We analyze the case with a constant
input-stream of requests and derive control-strategies for this.  This
is a simplified case it will constitute the foundation of adaptive
schemes to time-varying requests in the future.

We plan to extend this work by allowing for a dynamic input. It would
also be very natural to extend the model to account for uncertainties
in the service-rate. With this uncertainty it would be beneficial to
close a feedback loop around the service rate in order to guarantee a
desired service rate.

\paragraph{Acknowledgements}
The authors would like to thank Karl-Erik \AA{}rz\'en and Bengt
Lindoff for the useful comments on early versions of this technical
report.

\paragraph{Source code}
The source code used to compute the solution of the examples in
Section~\ref{sec:linear} and~\ref{sec:periodic} can be found
on Github at \url{https://github.com/vmillnert/REACTION-source-code}.

\bibliographystyle{IEEEtran}
\bibliography{service}

\end{document}